\documentclass[10pt,leqno]{amsart}
\usepackage{amsmath,amssymb,amscd,amsthm}
\usepackage{latexsym}
\usepackage{graphicx}
\usepackage{color}
\usepackage{bm}

\newtheorem{theorem}{Theorem}
\newtheorem{lemma}{Lemma}
\newtheorem{proposition}{Proposition}
\newtheorem{remark}{Remark}

\newtheorem{corollary}{Corollary}

\newtheorem{claim}{Claim}
\newtheorem{assumption}{Assumption}

\newcommand{\f}[2]{\frac{#1}{#2}}

\newcommand{\al}{\alpha}

\newcommand{\ga}{\gamma}

\newcommand{\de}{\delta}

\newcommand{\la}{\lambda}

\newcommand{\si}{\sigma}

\newcommand{\rone}{\mathbf R^1}

\newcommand{\cc}{\mathcal C}
\newcommand{\ch}{\mathcal H}
\newcommand{\calH}{\mathcal{H}}
\newcommand{\calL}{\mathcal{L}}
\newcommand{\calO}{\mathcal{O}}
\newcommand{\calP}{\mathcal{P}}

\newcommand{\rmd}{\mathrm{d}}
\newcommand{\rme}{\mathrm{e}}
\newcommand{\rmi}{\mathrm{i}}
\newcommand{\rmT}{\mathrm{T}}
\newcommand{\ess}{\mathrm{ess}}
\newcommand{\pt}{\mathrm{pt}}
\def\sech{\mathop\mathrm{sech}\nolimits}
\newcommand{\vI}{\bm{\mathit{I}}}
\newcommand{\vv}{\bm{\mathit{v}}}
\def\coloneqq{\mathrel{\mathop:}=}

\newcommand{\p}{\partial}

\newcommand{\beq}{\begin{equation}}
\newcommand{\eeq}{\end{equation}}
\newcommand{\beqna}{\begin{eqnarray*}}
\newcommand{\eeqna}{\end{eqnarray*}}
\newcommand{\beqn}{\begin{equation*}}
\newcommand{\eeqn}{\end{equation*}}
\newcommand{\bp}{\begin{proof}}
\newcommand{\ep}{\end{proof}}
\newcommand{\bprop}{\begin{proposition}}
\newcommand{\eprop}{\end{proposition}}
\newcommand{\bt}{\begin{theorem}}
\newcommand{\et}{\end{theorem}}
\newcommand{\bex}{\begin{Example}}
\newcommand{\eex}{\end{Example}}
\newcommand{\bc}{\begin{corollary}}
\newcommand{\ec}{\end{corollary}}
\newcommand{\bcl}{\begin{claim}}
\newcommand{\ecl}{\end{claim}}
\newcommand{\bl}{\begin{lemma}}
\newcommand{\el}{\end{lemma}}
\newcommand{\eps}{\epsilon}

\usepackage{color}

\begin{document}

\title
[ Spectral stability of kinks in $\mathcal{P T}$-symmetric Klein-Gordon]
{On the spectral stability of kinks in some $\mathcal{P T}$-symmetric variants  of  the classical  Klein-Gordon Field Theories }

\author{A. Demirkaya}

\author{M. Stanislavova }

\author{A. Stefanov}

\author{T. Kapitula}

\author{P.G. Kevrekidis  }

\address{Aslihan Demirkaya \\
Mathematics Department, \\
University of Hartford, \\
200 Bloomeld Avenue, West \\
Hartford CT 06112}
\email{demirkaya@hartford.edu}

\address{Milena Stanislavova, Atanas Stefanov \\
Department of Mathematics \\
University of Kansas\\
1460 Jayhawk Blvd\\ Lawrence, KS 66045--7523 }
\email{stefanov@math.ku.edu}

\address{Todd Kapitula \\
 Department of Mathematics and Statistics \\
 Calvin College \\
        Grand Rapids, MI 49546}
\email{tmk5@calvin.edu}

\address{Panayotis Kevrekidis \\
Lederle Graduate Research Tower\\
Department of Mathematics and  Statistics\\
University of Massachusetts\\
Amherst, MA 01003}
\email{kevrekid@math.umass.edu}

\thanks{Stefanov's research is supported in part by
 NSF-DMS 0908802.  Kevrekidis is supported by 
NSF-CMMI-1000337, NSF-DMS-1312856,
as well as by the US AFOSR under FA9550-12-1-0332,
the Binational Science Foundation through grant 2010239
and the FP7, Marie Curie Actions, People, International Research
Staff Exchange Scheme (IRSES-606096). }
\date{\today}

%\subjclass[2000]{ }

\keywords{Klein-Gordon PDEs, $\mathcal{PT}$-symmetry, kinks, spectral stability}

\begin{abstract}
In the present work we consider  the introduction of
$\mathcal{PT}$-symmetric terms in the context of classical Klein-Gordon
field theories. We explore the implication of such terms on the spectral
stability of coherent structures, namely kinks. We find that the conclusion
critically depends on the location of the kink center  relative to the
center of the $\mathcal{P T}$-symmetric term. The main result is that if
these two points coincide, the kink's spectrum remains on the imaginary
axis and the wave is spectrally stable. If the kink is centered on the
``lossy side'' of the medium, then it becomes stabilized. On the other
hand, if it becomes centered on the ``gain side'' of the medium, then it is
destabilized. The consequences of these two possibilities on the
linearization (point and essential) spectrum are discussed in some detail.
\end{abstract}

\maketitle

\section{Introduction}

In the last 15 years, since its inception in the context of
linear quantum mechanics~\cite{R1}, the study of systems with $\mathcal{P T}$
symmetry has gained tremendous momentum. While it was originally
proposed as a modification/extension of quantum mechanics,
its purview has gradually become substantially wider.
This stemmed to a considerable degree from the realization
that other areas such as optics might be ideally suited not
only for its theoretical study~\cite{ziad,Ramezani,Muga},
but also for its experimental realization~\cite{salamo,dncnat}.
Hence, this effort has motivated a wide range of explorations
focusing, among others, on the study of solitary waves
and breathers in lattice and continuum
systems~\cite{Abdullaev,baras2,baras1,malom1,malom2,Nixon,Dmitriev,Pelin1,Sukh},
as well as on the complementary aspect of low-dimensional (oligomer or
plaquette) settings~\cite{Li,Guenter,suchkov,Sukhorukov,ZK}.

 Most of the above developments, however, have focused on the
realm of Schr{\"o}dinger type operators. However, other classes
of Hamiltonian wave-bearing systems are of considerable interest
in their own right, with a notable example being Klein-Gordon equations.
 This arises not only in the context of field theories
(as e.g. in high energy/particle physics)~\cite{peyrard,dodd},
but also in that of discrete (lattice nonlinear dynamical) systems
with applications both to mechanics
(e.g., arrays of coupled pendula; see e.g. for a recent
example~\cite{anderson}) and to electrical systems (such as circuits
consisting of e.g. capacitive and inductive elements~\cite{remoiss}).
It is interesting to note here that both at the mechanical
level~\cite{bend_mech} and at the electrical one~\cite{R21,tsampas2}
realizations of $\mathcal{P T}$ symmetry and its breaking
have been recently implemented,
while a Klein-Gordon setting has also been explored theoretically
for so-called $\mathcal{P T}$ symmetric nonlinear metamaterials and
the formation of gain-driven discrete breathers therein~\cite{lazar}.

The aim of the present work is to formulate a prototypical field theoretic
setting where gain and loss coexist in a balanced $\mathcal{P T}$ symmetric
form in a Klein-Gordon model, around a central point which, without loss of
generality, will be assumed to be at $0$. This addition to the standard
Klein-Gordon model $u_{tt}=u_{xx} - V'(u)$ will come into the form of a
``dashpot'' term $\gamma(x) u_t$. However, contrary to the usual setting
where it has a negative definite prefactor incorporating dissipation (see
e.g. the relevant local dissipation term in~\cite{anderson}), here this term
will have a non-definite sign. In fact, to preserve the $\mathcal{P T}$
symmetry, $\gamma(x)$ will be anti-symmetric so that the transformation
$(x,t)\mapsto(-x,-t)$ will leave the equation invariant. This suggests that
there is a region where $\gamma(x)
>0$, representing (without loss of generality) the ``lossy'' side of the
medium, and a region where $\gamma(x)<0$, representing the ``gain'' side of
the medium. For definiteness, in our numerical simulations we will assume
that $\gamma(x)>0$ for $x<0$ (lossy side) and $\gamma(x)<0$ for $x>0$ (gain
side).

Our principal aim is to explore the impact of this term on the
stability properties of
the coherent structures of this equation. Arguably, the most important among
these are the ubiquitous kinks (heteroclinic connections in the ODE phase
space which connect two distinct fixed points). Of substantial interest are
also the relatively fragile breathers that exist in a time periodic
exponentially localized form for some models (such as the sine-Gordon
equation with $V(u)=1- \cos(u)$), but do not exist as such for others (such
as the $\phi^4$ model with $V(u)=u^2-\frac{1}{2} u^4$). Here, given the
fragility of the latter (breathers), we will focus solely on the former
(kinks). Moreover, even for the kinks we will consider the simpler stationary
problem, as the presence of $\gamma(x)$ and its breaking of translational
invariance destroys the Lorentz invariance which produces traveling solutions
out of static ones. Hence, it is not at all clear that traveling solutions
exist in the model in the presence of the $\mathcal{P T}$ symmetric
perturbation considered herein. Nevertheless, as the static problem $u_{xx} -
V'(u)=0$ is {\it unaffected} by our unusual antisymmetric dashpot, the steady
state (kink) solutions persist in the proposed model and are, in fact, still
available in explicit analytical form. For the sine-Gordon model they read
$\phi(x)=4 \arctan(\rme^{(x-x_0)})$, while for $\phi^4$, they are
$\phi(x)=\tanh(x-x_0)$. In each case, they can be centered at any $x_0$ along
the real line. The fundamental question is therefore one of their spectral
stability and it is that which we will try to resolve herein.

In particular, our mathematical analysis and numerical computations
lead to the following conclusions:
\begin{enumerate}
\item if $x_0=0$, i.e., the kink is centered at the interface between gain
    and loss, then the spectrum is unaffected (with the exception of a
    possible shift of point spectrum imaginary eigenvalues which can move
    along the imaginary spectral axis) and the kink is spectrally stable;
\item if the kink is centered on the ``lossy'' side of the medium, the kink
    will be spectrally stable (excluding the possibility of an eigenvalue
    with a very small real part), with one of the formerly neutral
    eigenvalues (associated with translation) moving to the left half plane;
\item if the kink is centered on the ``gain'' side of the medium, it will
    be spectrally unstable with one of the formerly neutral eigenvalues
    (associated with translation) moving to the right half plane.
\end{enumerate}

Our presentation is structured as follows. In section 2, we present the
models and their solutions, and we set up the relevant spectral problem. In
section 3 we prove a general result about the spectrum of a quadratic
eigenvalue problem which can be considered to be a generalization of the
problem associated with our two case examples. In section 4 we present
some numerical computations which illustrate the theory. We conclude in
section 5 with a brief discussion of our conclusions, and provide a possible
list of future questions to be addressed.

\section{Theoretical Setup and Principal Theorems}

We consider the following basic PDE
\begin{eqnarray}\label{5}
u_{tt}-u_{xx} +\eps \gamma(x) u_t   + f(u) = 0
\end{eqnarray}
Here $f(u)$ is a smooth nonlinear function that depends on the field $u$, and
$\gamma(x)$ is a function accounting for the presence of gain/loss in the
system. For $\eps = 0$, and for some class of  nonlinearities, there exist
static kink solutions $\phi$, which will be the main object of consideration
here. These will satisfy the following ODE
\begin{equation}\label{7}
-\phi_{xx}+f(\phi)=0.
\end{equation}
As mentioned in the introduction, we will be particularly interested in the
following examples:
\begin{itemize}
\item the sine-Gordon field theory, $f(u)=\sin(u)$, and the kink solution
    will be $\phi(x)=4 \arctan(\rme^x)$
\item the $\phi^4$ field theory, $f(u)=2 (u^3 - u)$, and the kink assumes
    the form $\phi(x)=\tanh(x)$.
\end{itemize}

The $\mathcal{P T}$ symmetry requires that the equation be invariant under
the mapping $(x,t)\mapsto(-x,-t)$ (spatial reflection and time reversal).
Clearly, the operator $\p_t^2-\p_x^2$ preserves the symmetry. In order for
the operator $\ga(x)\partial_t$ to also have this property, we need
$\gamma(x)$ to be an odd function. This will be assumed throughout the rest
of this paper.

Perturbing about the kink solution via the substitution
$u(t,x)=\phi(x)+v(t,x)$, we obtain the linearized PDE
\begin{equation}
\label{15}
v_{tt}-v_{xx} +\eps \ga(x) v_t   + f'(\phi) v=0.
\end{equation}
In our examples we have that
\[
\lim_{x\to \pm \infty} f'(\phi(x))=\sigma^2>0,
\]
so that $W(x)\coloneqq\sigma^2-f'(\phi(x))$ is an effective potential which
decays exponentially fast as $x\to\pm\infty$. We will henceforth assume that
$\gamma(x)$ is also exponentially localized. Introducing the Schr\"odinger
operator $\ch$ with domain $H^2(\rone)$ in the form
$$
\ch = -\p_x^2+f'(\phi)= -\p_x^2-W(x)+\sigma^2,
$$
we can rewrite the linearized problem as
\[
v_{tt}+\epsilon\gamma(x)v_t+\calH v=0.
\]
Upon making the transformation
\[
v(x,t)\mapsto v(x)\rme^{\lambda t}
\]
we have the spectral problem
\[
\lambda^2v+\epsilon\la \gamma(x)v+\calH v=0.
\]
If the real part of the eigenvalue is positive, the wave will be spectrally
unstable; otherwise, the wave is spectrally stable.

The Schr\"odinger operator $\ch$ is self-adjoint, and by Weyl's theorem its
essential spectrum is given by $\si_{\ess}(\ch)=[\sigma^2,+\infty)$.
Moreover, since solutions to the existence problem (\ref{7}) are invariant
under spatial translation, it will be the case that the operator has a
nontrivial kernel with $\ch(\phi_x)=0$. Alternatively, this fact follows
immediately upon differentiating \eqref{7} with respect to $x$. In the next
section we will stay away from the concrete $\ch$ that arises in our
application, and instead we will work with a general Schr\"odinger operator
which satisfies relevant properties.

\section{The quadratic eigenvalue problem}

We have properly motivated our study of the following spectral problem:
\begin{equation}\label{e:31}
\calP_2(\lambda)v\coloneqq\lambda^2v+\epsilon\lambda\gamma(x)v+\calH v=0.
\end{equation}
Here $\calH$ represents a self-adjoint Schr\"odinger operator on
$H^2(\mathbb{R})$,
\[
\calH\coloneqq-\partial_x^2-W(x)+\sigma^2,\quad\sigma>0,
\]
and $W(x)$ is a smooth function which decays exponentially fast,
\[
|W(x)|\le C\rme^{-\alpha|x|},\quad\alpha>0.
\]
We assume that $\gamma(x)$ is a smooth \textit{odd} function which satisfies
a similar exponential decay estimate as $W(x)$.

The exponential decay assumption on the potential $W(x)$ means that $\calH$
is a relatively compact perturbation of the constant coefficient operator
\[
\calH_\infty\coloneqq-\partial_x^2+\sigma^2
\]
(see \cite[Chapter~3.1]{Kap2} for the details). The essential spectrum of
$\calH$ is the spectrum of $\calH_\infty$, and is given by
\[
\sigma_{\ess}(\calH)=[\sigma^2,+\infty).
\]
Since $\calH$ is a relatively compact perturbation of an operator of Fredholm
index zero, the rest of the spectrum will be comprised of point eigenvalues,
and each of these will have finite algebraic multiplicity. Moreover, basic
Sturm-Liouville theory tells us that the point eigenvalues will be real and
simple (e.g., see \cite[Chapter~2.3]{Kap2}). Regarding the point spectrum of
the Schr\"odinger operator $\calH$, we make the following spectral
assumptions:

\begin{assumption}\label{ass:31}
The point spectrum of $\calH$, say $\sigma_0,\sigma_1^2,\dots,\sigma_n^2$, is
nonnegative and is ordered as
\[
0=\sigma_0<\sigma_1^2<\cdots<\sigma_n^2<\sigma^2.
\]
The associated eigenfunctions are labelled as $\psi_0,\psi_1,\dots,\psi_n$.
\end{assumption}

Quadratic eigenvalue problems are well-studied, and we refer to
\cite{Kap1,StanStef1,StanStef2,Pel1} for some relevant results and
references. One important result that we will use is that the quadratic
problem (\ref{e:31}) is equivalent to the linear eigenvalue problem
\begin{equation}\label{e:32}
\calL\left(\begin{array}{c}v\\w\end{array}\right)=
\lambda\left(\begin{array}{c}v\\w\end{array}\right),\quad
\calL\coloneqq\left(\begin{array}{rc}0&1\\-\calH&-\epsilon\gamma(x)\end{array}\right).
\end{equation}
As a consequence of this equivalence and the decay assumptions on $W(x)$ and
$\gamma(x)$, the operator $\calL-\lambda\vI_2$ is Fredholm of index zero for
all $\lambda$ which is not in the essential spectrum, $\sigma_{\ess}(\calL)$.
The essential spectrum is found by considering the spectrum of the asymptotic
operator
\[
\calL_\infty\coloneqq\left(\begin{array}{cc}0&1\\-\calH_\infty&0\end{array}\right)
\]
(the asymptotic operator is found by taking the limit $|x|\to\infty$ for
$\calL$). The essential spectrum is straightforward to compute, and we find
it is purely imaginary,
\[
\sigma_{\ess}(\calL)=(-\rmi\infty,-\rmi\sigma]\cup[\rmi\sigma,+\rmi\infty).
\]
Regarding the point spectrum, $\sigma_{\pt}(\calL)$, because $\calL$ is a
relatively compact perturbation of $\calL_\infty$ it will be the case that
there are only a finite number of such eigenvalues, and each has finite
algebraic multiplicity. Moreover, we have the symmetries as detailed below:

\begin{lemma}\label{le:31}
Suppose that $\epsilon\ge0$. If $W(x)$ is even, then the point spectrum
satisfies the Hamiltonian symmetry,
$\{\pm\lambda,\pm\overline{\lambda}\}\subset\sigma_{\pt}(\calL)$. If $W(x)$
is not even, then all that can be said is the point spectrum is symmetric
with respect to the real axis,
$\{\lambda,\overline{\lambda}\}\subset\sigma_{\pt}(\calL)$.
\end{lemma}

\begin{proof}
Setting $\vv=(v\,\,w)^\rmT$, the fact that $\calL$ has only real-valued
entries implies that
\[
\calL\vv=\lambda\vv\quad\Rightarrow\quad\calL\overline{\vv}=\overline{\lambda}{\overline{\vv}}.
\]
Thus, the spectral symmetry with respect to the real axis is established. Now
suppose that $W(x)$ is even. Setting $y=-x$ gives
\[
\gamma(-y)=-\gamma(y),\quad W(-y)=W(y),
\]
so that the quadratic problem becomes
\[
\lambda^2z-\epsilon\lambda\gamma(y)z+\calH z=0,\quad z(y)=v(-y).
\]
The Schr\"odinger operator $\calH$ is invariant under the spatial reflection.
Upon setting $\gamma=-\lambda$ we recover the original eigenvalue problem,
i.e.,
\[
\calP_2(\gamma)z=0.
\]
Thus, under the assumption that $W(x)$ is even we see that if $\lambda$ is an
eigenvalue with associated eigenfunction $v(x)$, then $-\lambda$ is an
eigenvalue with associated eigenfunction $v(-x)$.
\end{proof}

Since the origin is an eigenvalue for $\calH$ with eigenfunction $\psi_0$, it
will be an eigenvalue for $\calL$ with associated eigenfunction
$(\psi_0\,\,0)^\rmT$. Moreover, since $\calH$ is a Schr\"odinger operator, it
is clear that the geometric multiplicity of the zero eigenvalue is one.
Regarding the algebraic structure of this eigenvalue we have:

\begin{lemma}\label{le:5}
Suppose that $\epsilon>0$ is sufficiently small. The eigenvalue
$\{0\}\subset\sigma_{\pt}(\calL)$ has algebraic multiplicity two if
\[
\langle\gamma(x)\psi_0,\psi_0\rangle=0,\quad
\langle f,g\rangle\coloneqq\int_{-\infty}^{+\infty}f(x)g(x)\,\rmd x.
\]
Otherwise, it has algebraic multiplicity one.
\end{lemma}

\begin{proof}
The generalized kernel is found by solving
\[
\calL\left( \begin{array}{c} v \\ w \end{array}\right)=
\left( \begin{array}{c} \psi_0 \\ 0 \end{array}\right),
\]
which with $w=\psi_0$ gives
\[
\calH v = -\epsilon\gamma(x)\psi_0.
\]
By the Fredholm alternative this equation can be solved if and only if
$\gamma(x)\psi_0\in\mathrm{Ker}(\calL)^\perp$, i.e.,
\[
\langle\gamma(x)\psi_0,\psi_0\rangle=0.
\]

Supposing that the algebraic multiplicity is at least two, we will now show
that the Jordan chain cannot be any longer. We start with
\[
\calL\left(\begin{array}{c} -\epsilon\calH^{-1}[\gamma(x)\psi_0] \\ \psi_0 \end{array}\right)=
\left( \begin{array}{c} \psi_0 \\ 0 \end{array}\right).
\]
If the chain is longer, then we can solve the linear system
\[
\calL\left( \begin{array}{c} v \\ w \end{array}\right)=
\left(\begin{array}{c} -\epsilon\calH^{-1}[\gamma(x)\psi_0] \\ \psi_0 \end{array}\right),
\]
which with $w=-\epsilon\calH^{-1}[\gamma(x)\psi_0]$ gives
\[
\calH v=-\psi_0+\epsilon^2\gamma(x)\calH^{-1}[\gamma(x)\psi_0].
\]
Again invoking the Fredholm alternative, we see that this system has a
solution if and only if
\[
\begin{split}
\langle\psi_0,\psi_0\rangle&=\epsilon^2\langle\gamma(x)\calH^{-1}[\gamma(x)\psi_0],\psi_0\rangle\\
&=\epsilon^2\langle\calH^{-1}[\gamma(x)\psi_0],\gamma(x)\psi_0\rangle.
\end{split}
\]
By the spectral assumption on $\calH$, i.e., all of the eigenvalues are
nonnegative, as well as the assumption that
$\langle\gamma(x)\psi_0,\psi_0\rangle=0$, we know that
\[
\langle\calH^{-1}[\gamma(x)\psi_0],\gamma(x)\psi_0\rangle
%\ge\frac{|\langle\gamma(x)\psi_0,\psi_1\rangle|^2}{\sigma_1^2}
>0.
\]
Thus, the equation can be guaranteed to not have a solution only if
$\epsilon>0$ is not too large.
\end{proof}

\begin{remark}
If $\epsilon=0$, then there is a Jordan chain of length two, and the
generalized eigenfunction is given by $(0\,\,\psi_0)^\rmT$.
\end{remark}

Now that we understand the spectral structure at the origin, we consider the
question of whether or not there can be eigenvalues which are arbitrarily
large in magnitude with nonzero real part. Before doing so, let us consider
the structure of the spectrum when $\epsilon=0$. The eigenvalue problem
becomes
\[
\calH v=-\lambda^2v,
\]
so that the eigenvalues for $\calL$ are intimately related to those for
$\calH$. Indeed, as a consequence of the spectral assumption on $\calH$, we
see that there are precisely $n$ purely imaginary simple eigenvalues with
positive imaginary part
\[
\lambda_j=\rmi\sigma_j,\quad j=1,\dots,n.
\]
The spectral symmetry of Lemma \ref{le:31} guarantees that there are $n$
purely imaginary eigenvalues with negative imaginary part. An associated
eigenfunction for each eigenvalue with positive imaginary part is
$(\psi_j\,\,\rmi\psi_j)^\rmT$. Thus, arbitrarily large eigenvalues with
nonzero real part for $\epsilon>0$ can exist only if they are somehow ejected
from the essential spectrum.

Before continuing, we shall need the following technical result, which is
sometimes called {\it limited absorption principle}. The limited  absorption
principle is a  standard result, whose proof is based on appropriate
estimates for the Jost solutions of the corresponding operator $L$.  Its
proof goes back to the fundamental work of~\cite{Agmon}, which
was extended later by~\cite{KatoJensen} and~\cite{DeiftTrub}.
%Kato-Jensen, \cite{KatoJensen}
%Deift-Trubowitz, \cite{DeiftTrub}.

\begin{lemma}\label{le:40}
Consider the standard weighted $L^2$ space defined by the norm
\[
\|f\|_{L^2_\alpha}\coloneqq\left( \int_{-\infty}^{+\infty}
  |f(x)|^2 (1+|x|^2)^{\al}\,\rmd x\right)^{1/2}=:\| <x>^{\al} f\|_{L^2}.
\]
For every $\de>0$ there exists $C=C_{\de, W}$ so that for all $\la\in \cc$
with%\footnote{it is well-known that such Schr\"odinger operators $L$ will
%have at most finitely many eigenvalues}
\[
\mathrm{dist}(\la,\si_{\pt}(\calH))>\de,
\]
we have the estimate
\begin{equation}\label{eq:po}
\|(\calH-\la^2)^{-1}\|_{L^2_{\f{1}{2}+\de}\to L^2_{-\f{1}{2}-\de}}\leq
\f{C}{|\la|}.
\end{equation}
%If in addition, we assume that $0$ is neither an eigenvalue nor a resonance
%%\footnote{By this we mean that the equation $-u''+W u=0$ \underline{does  not}  have distributional solutions in $L^\infty(\rone)$, as in \eqref{m:10}.  }
%for $L$, then  for all $\la\in \cc: dist(\la, \si_{p.p.}(L))>\de$
%\begin{equation}
%\label{80}
%\|(L-\la^2)^{-1}\|_{L^2_{\f{3}{2}+\de}\to L^2_{-\f{3}{2}-\de}}\leq \f{C_\de}{1+|\la|},
%\end{equation}
%that is \eqref{eq:po} extends to the case $\la=0$ (at the expense of taking
%slightly larger weights in the $L^2$ spaces).
\end{lemma}

We are now ready to show that if $\epsilon>0$ is not too large, then there
are no arbitrarily large eigenvalues.

\begin{lemma}\label{prop:20}
There exists an $\epsilon_0>0$ and $M>0$ such that if
$0\le\epsilon<\epsilon_0$, then there exist no $\lambda\in\sigma_{\pt}(\calL)$
with nonzero real part such that $|\lambda|>M$.
\end{lemma}

\begin{proof}
We prove by contradiction. Assume there is a sequence of point eigenvalues
with nonzero real part $\la_n\to \infty$ corresponding to a sequence
$\eps_n\to 0$. Let $\phi_n$ be a corresponding sequence of normalized
eigenfunctions, $\|\phi_n\|=1$, so that $\calP_2(\lambda_n)\phi_n=0$, i.e.,
\begin{equation}\label{60}
(\calH+\la_n^2)\phi_n=-\eps_n \la_n \ga(x) \phi_n.
\end{equation}
Since $  \la_n\to \infty$,  we can apply the limited  absorption principle in
Lemma \ref{le:40} to conclude that $\calH+\la_n^2$ is invertible (at least in
the sense of the weighted $L^2$ spaces as specified there), so
\begin{equation}
\label{50a}
\phi_n= -\eps_n \la_n (\calH+\la_n^2)^{-1}[\ga(x) \phi_n].
\end{equation}
Fixing $\de=1/2$ (for example) and taking $L^2_{-1}$ norms in \eqref{50a}
yields
\[
\begin{split}
\|\phi_n\|_{L^2_{-1}} &\leq \f{C}{|\la_n|}  \eps_n |\la_n| \|\ga(x) \psi_n\|_{L^2_{1}}\\
&\leq C \eps_n\left(\sup_x\left[<x>^2 |\ga(x)|\right]\right) \|\phi_n\|_{L^2_{-1}}.
\end{split}
\]
The last inequality leads to a contradiction, for $\eps_n\to0$ as
$n\to\infty$.
\end{proof}

We now consider the problem of whether or not \textit{any} eigenvalues can be
ejected from the essential spectrum when $\epsilon>0$. Alternatively, one can
think of the following as a discussion on the resonances of a Schr\"odinger
operator. Here we rely on the discussion of \textit{edge bifurcations} for
Schr\"odinger operators as presented in \cite[Chapter~9.5]{Kap2}. Recalling
the spectral symmetry about the real axis, assume that the imaginary part of
$\lambda$ is positive. Upon setting
\[
\mu^2\coloneqq\sigma^2+\lambda^2\quad\Rightarrow\quad
\lambda=\rmi\sqrt{\sigma^2-\mu^2},
\]
the eigenvalue problem $\calP_2(\lambda)v=0$ can be rewritten as
\[
v_{xx}+W(x)v=\mu^2v+\rmi\epsilon\sqrt{\sigma^2-\mu^2}\,\gamma(x)v.
\]
An important point here is that the branch point $\lambda=\rmi\sigma$ has
been transformed to $\mu=0$. The transformation defines a two-sheeted Riemann
surface in a neighborhood of the branch point. The physical sheet of this
transformation, $\mu\in(-\pi/2,\pi/2)$, corresponds to eigenvalues for the
original problem, while the resonance sheet, $\mu\in(\pi/2,3\pi/2)$,
corresponds to resonance poles for the original problem. In a neighborhood of
this branch point one can construct an analytic Evans function, $E(\mu)$,
which has the property $E(\mu)=0$ if $\mu$ is an eigenvalue (physical sheet),
or if $\mu$ is a resonance pole (resonance sheet).

First suppose that $\epsilon=0$. If $W(x)$ is a \textit{reflectionless}
potential, i.e., if $\calH$ has a resonance at the edge of the essential
spectrum, then there is a uniformly bounded but nondecaying solution, $\Psi$,
at the branch point of $\calH$,
\begin{equation}\label{e:bp2}
\calH(\Psi)=\sigma^2\Psi,\quad\Psi(x)\to\begin{cases}
\,\,1,\quad&x\to-\infty\\
\Psi_+,\quad&x\to+\infty.
\end{cases}
\end{equation}
The Evans function notes the existence of this ``eigenfunction" via $E(0)=0$
with $E'(0)\neq0$. If there is no resonance, then $E(0)\neq0$. Since the
Evans function is analytic, in a neighborhood of the branch point it will
have no other possible zeros.

Now suppose that $\epsilon>0$ is small. The zero at the branch point will
move to either the physical sheet or the resonance sheet, and in either case
it will be of $\calO(\epsilon)$. Denote this zero by $\mu_1\epsilon$, where
$\mu_1=\mu_1(\epsilon)$. It is shown in \cite[Theorem~9.5.1]{Kap2} that the
expansion for this zero as applied to this problem is given by
\begin{equation}\label{e:bp}
\mu_1=-\rmi\sigma\frac{\langle\gamma(x)\Psi,\Psi\rangle}{1+\Psi_+^2}+\calO(\epsilon).
\end{equation}
Here $\Psi$ is given in \eqref{e:bp2}. Note that to leading order this zero
is purely imaginary on the Riemann surface, and is consequently on the
boundary between the physical sheet and the resonance sheet. If it moves onto
the physical sheet, it will create an eigenvalue given by
\[
\lambda=\rmi\sqrt{\sigma^2-\mu_1^2\epsilon^2}=\rmi\left(\sigma+\calO(\epsilon^2)\right).
\]
The eigenvalue created by the edge bifurcation will be $\calO(\epsilon^2)$
close to the branch point. Because the leading order term in the expansion of
$\mu_1(\epsilon)$ is purely imaginary, the $\calO(\epsilon^2)$ term in the
expansion is purely real. Thus, if the eigenvalue has a nonzero real part it
will necessarily be of $\calO(\epsilon^3)$.  If the zero moves onto the
resonance sheet, no eigenvalue is created. On the other hand, if the
potential is not reflectionless, then no zeros will come out of the branch
point. Finally, the Evans function analysis also shows that an eigenvalue can
emerge from the essential spectrum only at the branch point.

\begin{lemma}\label{lem:bp}
Consider the spectral problem near the branch point $\lambda=\rmi\sigma$. If
the potential $W(x)$ is reflectionless, then for $\epsilon>0$ small it is
possible that there will be an eigenvalue which is $\calO(\epsilon^2)$ close
to the branch point. The real part of the eigenvalue will be of
$\calO(\epsilon^3)$. If the potential is not reflectionless, then no such
eigenvalue can emerge. In either case there will be no other eigenvalues
which have imaginary part greater than $\sigma+\calO(\epsilon^2)$.
\end{lemma}

\begin{figure}[tbp]
\includegraphics{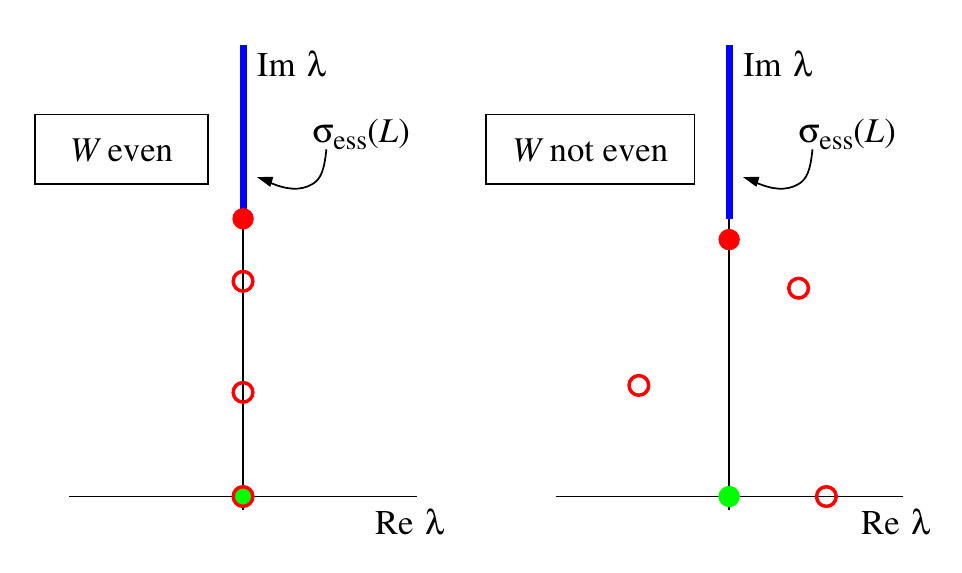}
\caption{A cartoon which illustrates the possible spectrum for small
$\epsilon$ when $W(x)$ is even (left panel) and $W(x)$ is not even (right panel). As a
consequence of the spectral symmetry with respect to the real axis, only the upper-half
of the complex plane is shown. The essential spectrum is given by a thick (blue) line, and
the eigenvalue associated with a potential edge bifurcation is given by a filled (red)
circle. The other nonzero eigenvalues are denoted by open circles. The double eigenvalue at
the origin when $\epsilon=0$ is denoted by a filled circle. When $W(x)$ is even the eigenvalues
remain purely imaginary, and the eigenvalue at the origin remains a double eigenvalue,
whereas if $W(x)$ the eigenvalues will generically have nonzero
real part, and the eigenvalue at the origin becomes simple. The potential simple eigenvalue
arising from the edge bifurcation remains close to the edge of the essential spectrum.}
\label{f:spectralcartoon}
\end{figure}

We are now in position to give a full description of the spectrum for small
$\epsilon>0$. The spectrum for $\epsilon=0$ is fully understood: it is purely
imaginary, has a finite number of nonzero simple eigenvalues in the gap
$(-\rmi\sigma,+\rmi\sigma)$, and has a geometrically simple eigenvalue at the
origin with algebraic multiplicity two. For $\epsilon>0$ each of the nonzero
eigenvalues will move a distance of $\calO(\epsilon)$. Indeed, upon writing
\[
\calL=\left(\begin{array}{rc}0&1\\-\calH&0\end{array}\right)+
\epsilon\left(\begin{array}{rc}0&0\\0&-\gamma(x)\end{array}\right),
\]
and using the fact an eigenfunction associated to the eigenvalue
$\rmi\sigma_j$ is $(\psi_j\,\,\rmi\psi_j)^\rmT$, we see that after using
standard perturbation theory (see \cite[Chapter~6]{Kap2}) each of these
eigenvalues will have the expansion
\[
\lambda_j=\rmi\sigma_j-
\epsilon\frac{\langle\gamma(x)\psi_j,\psi_j\rangle}{2\langle\psi_j,\psi_j\rangle}+
\calO(\epsilon^2).
\]
In particular, if $\langle\gamma(x)\psi_j,\psi_j\rangle\neq0$, then the
eigenvalue will move off the imaginary axis and have nonzero real part. As
for the double eigenvalue at the origin, if
$\langle\gamma(x)\psi_0,\psi_0\rangle\neq0$ one of the two eigenvalues will
leave. The nonzero eigenvalue will be of $\calO(\epsilon)$, and will have the
expansion (again see \cite[Chapter~6]{Kap2})
\[
\lambda_0=-\epsilon\frac{\langle\gamma(x)\psi_0,\psi_0\rangle}{\langle\psi_j,\psi_j\rangle}+
\calO(\epsilon^2).
\]
To leading order the nonzero eigenvalue will be purely real. The spectral
symmetry with respect to the real axis guarantees that the nonzero eigenvalue
is indeed real.

We are now ready to state our main result. The key control of the spectral
structure is whether or not the potential $W(x)$ is an even function. The
cartoon in Fig. \ref{f:spectralcartoon} gives an illustration of the result.

\begin{theorem}\label{thm:31}
Consider the quadratic spectral problem
\[
\lambda^2v+\epsilon\la \gamma(x)v+\calH v=0,
\]
where $\calH$ is the Schr\"odinger operator
\[
\calH=-\partial_x^2-W(x)+\sigma^2.
\]
Assume that $W(x),\gamma(x)$ are exponentially localized, that $\gamma(x)$ is
an odd function, and that $\epsilon>0$ is sufficiently small. Further suppose
that $\calH$ satisfies the spectral Assumption \ref{ass:31}. If $W(x)$ is
even, then the spectrum is purely imaginary, and the eigenvalue at the origin
will have algebraic multiplicity two. If $W(x)$ is not even, the origin will
be a simple eigenvalue, and there will be a purely real eigenvalue which has
the expansion
\begin{equation}\label{eigv2}
\lambda_0=-\epsilon\frac{\langle\gamma(x)\psi_0,\psi_0\rangle}{\langle\psi_0,\psi_0\rangle}+
\calO(\epsilon^2).
\end{equation}
Moreover, in the upper-half plane there will be $n$ eigenvalues, each of
which has the expansion
\begin{equation}\label{eigv3}
\lambda_j=\rmi\sigma_j-
\epsilon\frac{\langle\gamma(x)\psi_j,\psi_j\rangle}{2\langle\psi_j,\psi_j\rangle}+
\calO(\epsilon^2).
\end{equation}
In particular, the real part will be $\calO(\epsilon)$. Finally, there may be
a simple eigenvalue near the branch point, but if so its real part will be
$\calO(\epsilon^3)$.
\end{theorem}

\begin{proof}
Suppose that $W(x)$ is even. For any $\epsilon\ge0$ the spectrum satisfies
the Hamiltonian symmetry $\{\pm\lambda,\pm\overline{\lambda}\}$. This
symmetry first implies that if an eigenvalue is simple and purely imaginary
for $\epsilon=0$, then it must remain so for $\epsilon>0$ sufficiently small.
Indeed, it cannot leave the imaginary axis until it collides with another
purely imaginary eigenvalue. If $W(x)$ is reflectionless, then it is possible
for one eigenvalue to emerge from the essential spectrum through the edge
bifurcation. The Hamiltonian spectral symmetry forces this eigenvalue to be
purely imaginary for small $\epsilon$. Finally, the eigenvalue at the origin
must remain a double eigenvalue. One eigenvalue must always remain at the
origin due to $\calH\psi_0=0$, and we know the other is purely real. But, the
Hamiltonian symmetry implies that real eigenvalues come in pairs (symmetric
about the imaginary axis). Since this scenario is precluded, the eigenvalue
cannot leave.

On the other hand, if $W(x)$ is not even, then we have the perturbation
expansions as proven before the statement of the theorem.
\end{proof}

\begin{remark}
If $W(x)$ is even, then the eigenfunctions of the Schr\"odinger operator
$\calH$ are either even or odd. The assumption that the smallest eigenvalue
of $\calH$ is zero implies via Sturm-Liouville theory that the associated
eigenfunction is even and nonzero. An eigenfunction associated with the
eigenvalue $\sigma_j^2$ will be odd if $j$ is odd, and even if $j$ is even;
moreover, it will have $j$ zeros. In any event, since $\gamma(x)$ is odd it
will be the case that the first-order terms in the perturbation expansions
will be zero; hence, if the nonzero eigenvalues move, they will do so at a
rate of $\calO(\epsilon^2)$.
\end{remark}

\section{Numerical Results}

We now turn to numerical considerations in order to identify the motion of
the relevant eigenvalues and corroborate the predictions of Theorem
\ref{thm:31}. We first need to know the spectrum associated with the
Schr\"odinger operator $\calH$ for the sine-Gordon model and the $\phi^4$
model. We will henceforth assume that
\[
\gamma(x)=-x\rme^{-x^2/2},
\]
so that $\gamma(x)>0$ for $x<0$ (lossy side), and $\gamma(x)<0$ for $x>0$
(gain side).

First consider the sine-Gordon model. Since
\[
f'(\phi)=\sin(\phi),\quad\phi(x)=4\arctan(\rme^x),
\]
we can write $\calH$ as
\[
\calH=-\partial_x^2+\cos(\phi)=-\partial_x^2+(1-\cos(\phi))+1.
\]
The branch point in the upper-half plane is $\lambda_{\mathrm{br}} = \rmi$.
Since there is a bounded eigenfunction for $\calH$ at the branch point
$\sigma^2=1$,
\[
\calH(\Psi)=\Psi,\quad\Psi(x)=\tanh(x),
\]
an Evans function analysis yields that the potential $W(x)=\cos(\phi)-1$ is
reflectionless. Thus, it is possible for an eigenfunction to emerge from the
branch point $\lambda_{\mathrm{br}}$. Other than the zero eigenvalue, there
are no nonnegative eigenvalues. According to Theorem \ref{thm:31} the absence
of point spectrum eigenvalues other than the double one at the origin
suggests that to leading order this is the only eigenvalue that one needs to
worry about under the effect of the $\mathcal{P T}$ symmetric term. The
eigenvalue which arises from the origin will be real and of
$\calO(\epsilon)$, whereas the eigenvalue which potentially arises from the
edge bifurcation will have a real part of $\calO(\epsilon^3)$.

Now consider the $\phi^4$ model. Since
\[
f'(\phi)=-2+6\phi^2,\quad\phi(x)=\tanh(x),
\]
we can write $\calH$ as
\[
\calH=-\partial_x^2-2+6\phi^2=-\partial_x^2+6\sech^2(x)+4.
\]
Thus, $\sigma^2=4$, so the branch point in the upper-half plane is
$\lambda_{\mathrm{br}}=\rmi2$. Regarding the point spectrum, the problem with
$W(x)=-6\sech^2(x)$ was considered in \cite[Chapter~9.3.2]{Kap2}; see
also~\cite{sugiy}. Therein it was shown that this is a reflectionless
potential, which implies that it is possible for an eigenvalue to emerge from
the branch point $\lambda_{\mathrm{br}}$. Moreover, in addition to the zero
eigenvalue there is one positive eigenvalue given by $\sigma^2=3$. According
to Theorem \ref{thm:31} there are two eigenvalues of concern: the one near
the origin, and the one near $\rmi\sqrt{3}$. Under the effect of the
$\mathcal{P T}$ symmetric term there will again be a purely real eigenvalue
of $\calO(\epsilon)$, and unlike the sine-Gordon model there will also be an
eigenvalue near $\rmi\sqrt{3}$ which has a real part of $\calO(\epsilon)$. If
an edge bifurcation occurs, the bifurcating eigenvalue will have a real part
of $\calO(\epsilon^3)$.

\begin{figure}[tbp]
\includegraphics[width=8cm]{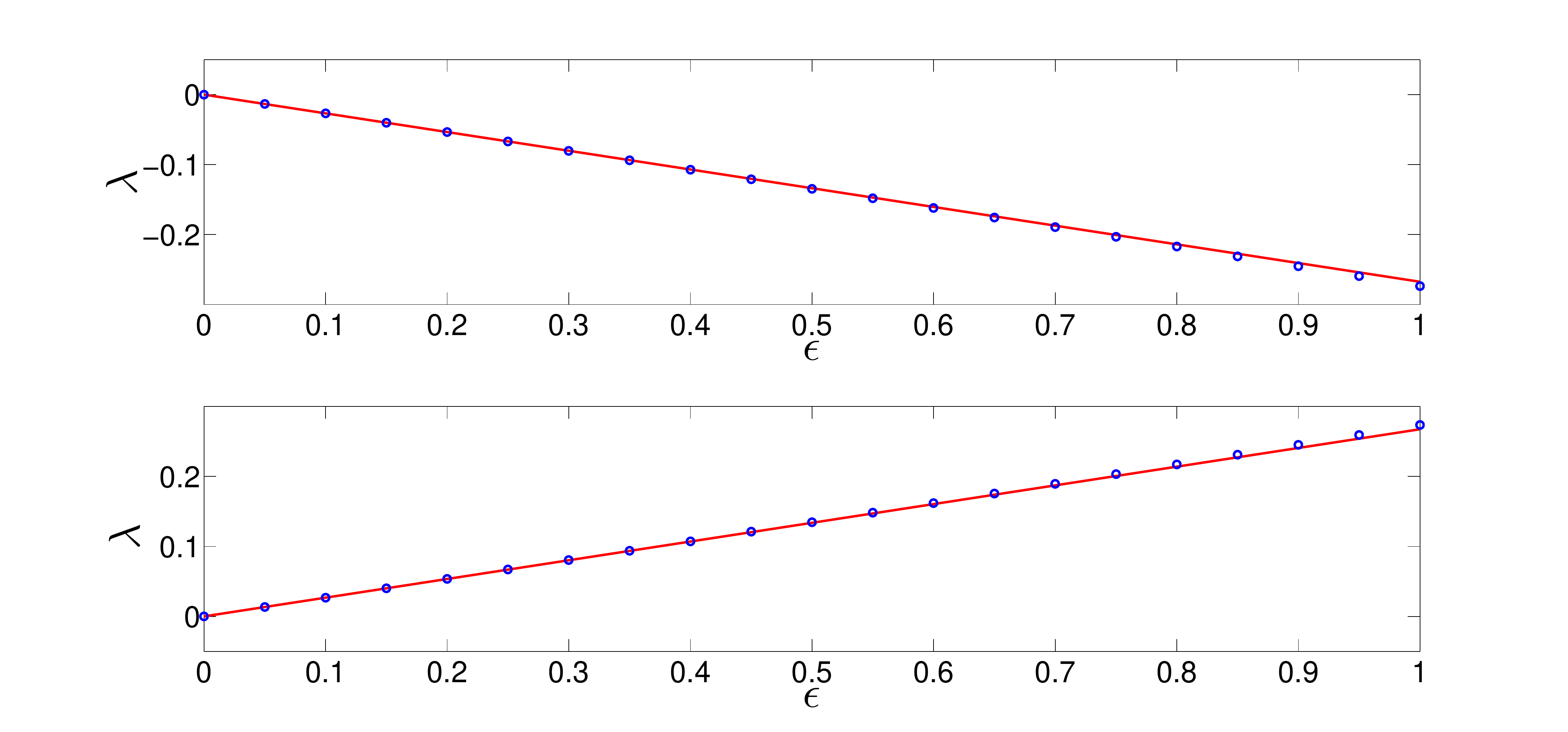}
\caption{The shifting of the formerly neutral eigenvalue
in the presence of a $\mathcal{P T}$ symmetric perturbation
for the sine-Gordon model, as described in the text for
a kink centered on the lossy side ($x_0=-2$, top) and
one centered on the gain side ($x_0=2$, bottom). The dots
denote numerical results for different values of $\epsilon$
(prefactor of the $\mathcal{P T}$ symmetric term), while the
solid line yields the analytical prediction of Eq.~(\ref{eigv2})
[see also the text]. Only for $\epsilon \approx 1$, does the
numerical result start to weakly deviate from the explicit theoretical
approximation.}
\label{fig1}
\end{figure}

The important point to note here is that each equation is invariant under
spatial translation, so that a translate of a solution is also a solution. If
we consider a solution which has been translated, then it the case that the
eigenfunctions will be translated by the same amount. Translating the wave,
i.e., moving its center, will act as a way to break symmetry. Moving the
center will force the initially purely imaginary eigenvalues associated with
the linearization about the centered wave (the wave balanced between the gain
side and the lossy side) to gain a nontrivial real part.

Let us first consider the possible bifurcation from the origin for each
model. For the sine-Gordon model we have
\[
\phi(x-x_0)=4\arctan(\rme^{x-x_0})\quad\Rightarrow\quad
\psi_0(x)=\sech(x-x_0).
\]
The associated potential $W(x)$ is even for $x_0=0$; otherwise, it is not. We
know that the spectrum will be purely imaginary if $x_0=0$: the question is
what happens for $x_0\neq0$. A direct numerical calculation of the integrals
associated with the perturbation expansion shows that for the eigenvalue
bifurcating from the origin,
\[
\lambda\sim\begin{cases}\,\,\,\,0.2675\epsilon,\quad&x_0=2\\
-0.2675\epsilon,\quad&x_0=-2
\end{cases}
\]
(see Fig. \ref{fig1}). In other words, if the underlying wave is centered on
the gain side, then there will be a real positive eigenvalue and the wave
will be unstable, whereas if the wave is centered on the loss side the
bifurcating real eigenvalue will be negative. If it were the case that the
potential eigenvalue emerging from the branch point $\lambda=\rmi$ also has
nonpositive real part, then the wave would be spectrally stable. While we
will not show the details here, a preliminary calculation suggests that if
there is an eigenvalue which arises from an edge bifurcation, then the real
part must be $\calO(\epsilon^3)$.

\begin{figure}[tbp]
\includegraphics[width=8cm]{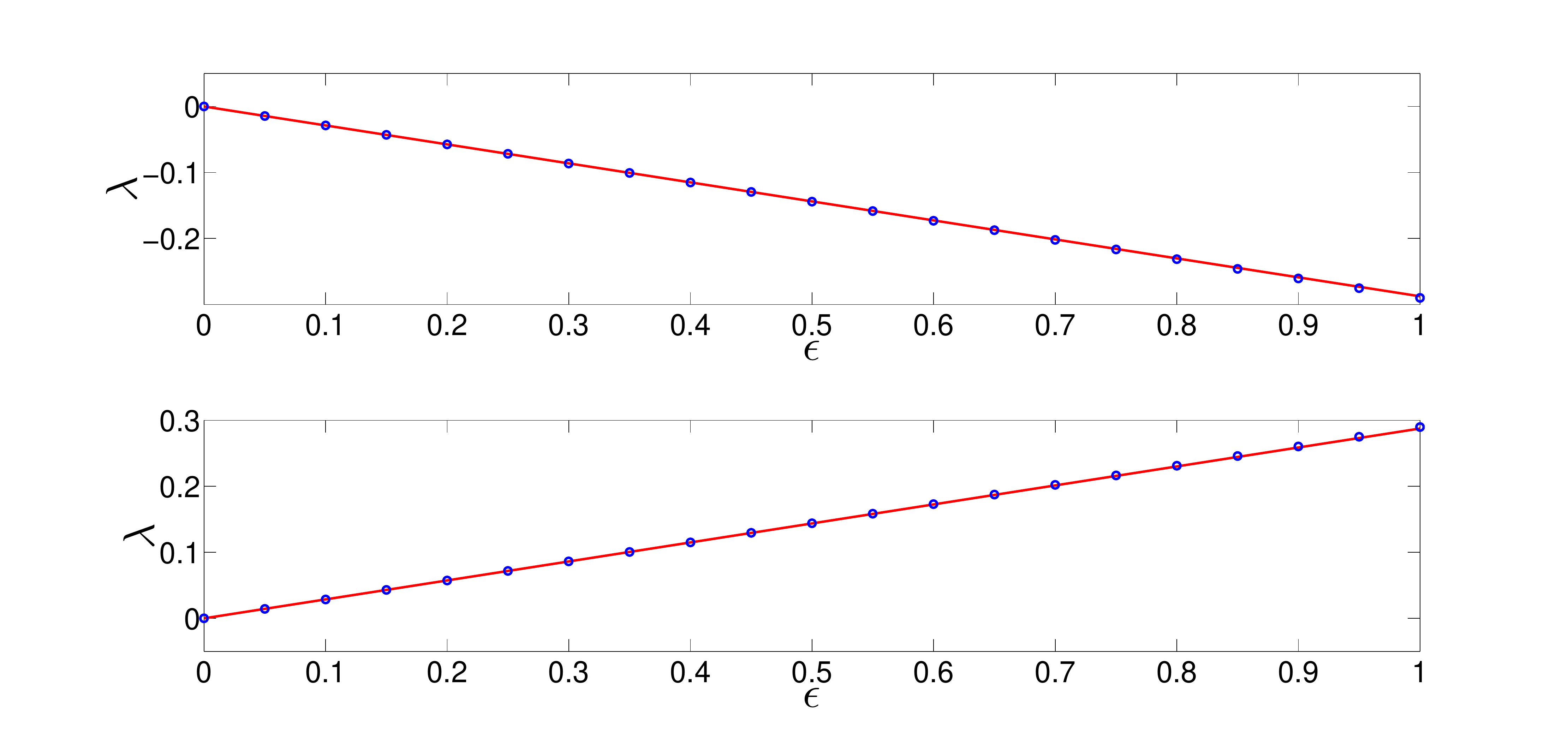}
\caption{The same as the previous figure but now for
the $\phi^4$ model.}
\label{fig2}
\end{figure}

For the $\phi^4$ model we have
\[
\phi(x-x_0)=\tanh(x-x_0)\quad\Rightarrow\quad
\psi_0(x)=\sech^2(x-x_0).
\]
The associated potential $W(x)$ is even for $x_0=0$; otherwise, it is not. We
know that the spectrum will be purely imaginary if $x_0=0$: the question is
what happens for $x_0\neq0$. A direct numerical calculation of the integrals
associated with the perturbation expansion shows that for the eigenvalue
bifurcating from the origin,
\[
\lambda\sim\begin{cases}\,\,\,\,\,0.2876\epsilon,\quad&x_0=+2\\
-0.2876\epsilon,\quad&x_0=-2
\end{cases}
\]
(see Fig. \ref{fig2}). Again, if the underlying wave is centered on the gain
side, then there will be a real positive eigenvalue and the wave will be
unstable, whereas if the wave is centered on the loss side the bifurcating
real eigenvalue will be negative. There is the additional eigenvalue which
starts at $\lambda=\rmi\sqrt{3}$. Since
\[
\psi_1(x)=\tanh(x)\sech(x)
\]
(see \cite{sugiy}), a numerical calculation of the integrals gives
\[
\lambda\sim\rmi\sqrt{3}+\begin{cases}\,\,\,\,\,0.1137\epsilon,\quad&x_0=+2\\
-0.1137\epsilon,\quad&x_0=-2\end{cases}
\]
(see Fig. \ref{fig3} for the case of $x_0=-2$, and Fig. \ref{fig4} for the
case of $x_0=2$). As is the case with the eigenvalue bifurcating from the
origin, if the wave is centered on the gain side, the real part of the
eigenvalue will be positive, whereas if the wave is centered on the lossy
side, the real part will be negative. If it were the case that the potential
eigenvalue emerging from the branch point $\lambda=\rmi2$ also has
nonpositive real part, then the wave would be spectrally stable. While we
will not show the details here, a preliminary calculation again suggests that
if there is an eigenvalue which arises from an edge bifurcation, then the
real part must be $\calO(\epsilon^3)$.

\begin{figure}[tbp]
\includegraphics[width=8cm]{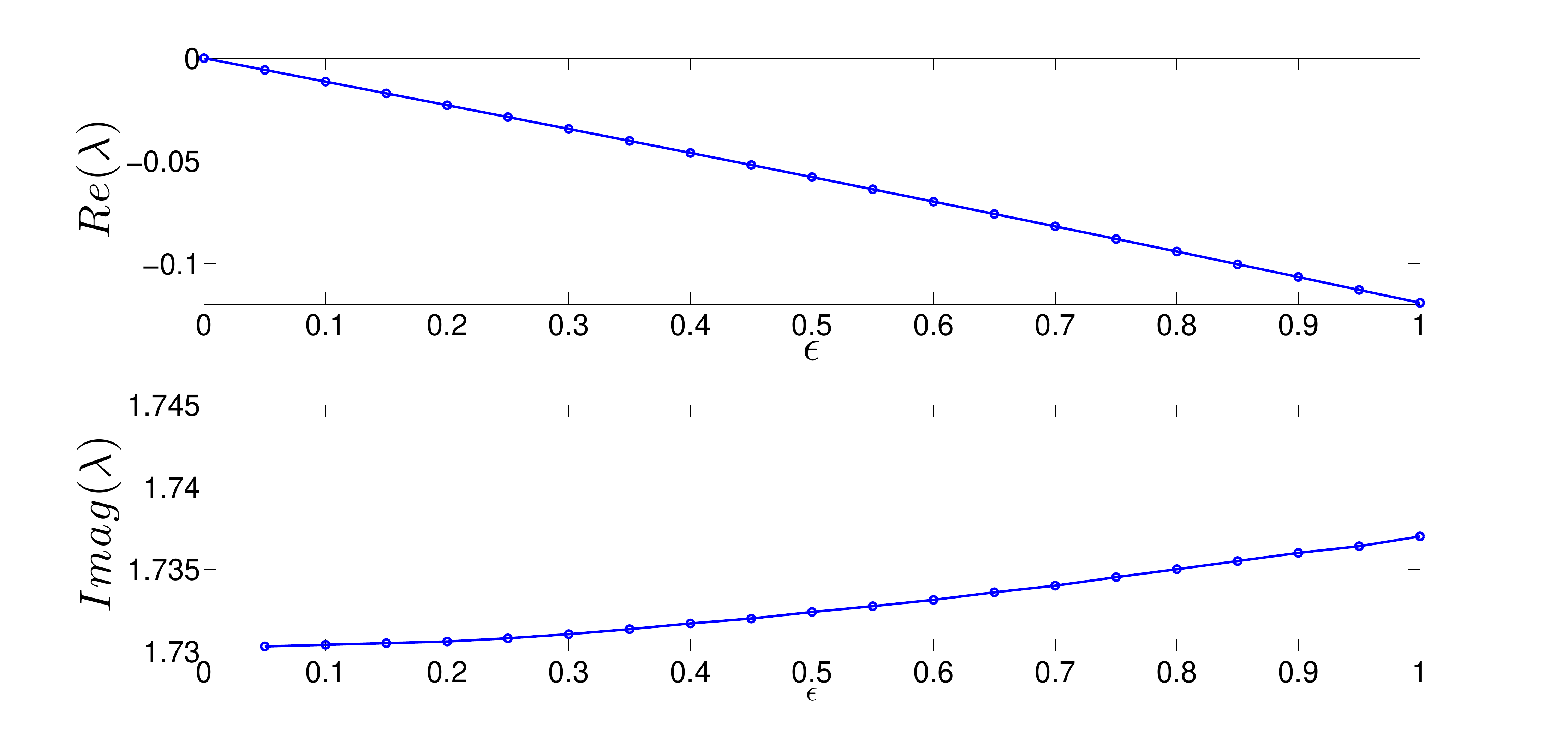}
\caption{The figure shows the dependence of the real (top) and imaginary
(bottom) part of the eigenvalue originally (i.e., for $\epsilon=0$)
located at $\lambda^2=-3$ for the $\phi^4$ model. For the kink centered
at $x_0=-2$, the perturbation knocks this pair off of the imaginary axis
and into a complex pair in the left half plane, implying spectral
stability.}
\label{fig3}
\end{figure}

\begin{figure}[tbp]
\includegraphics[width=8cm]{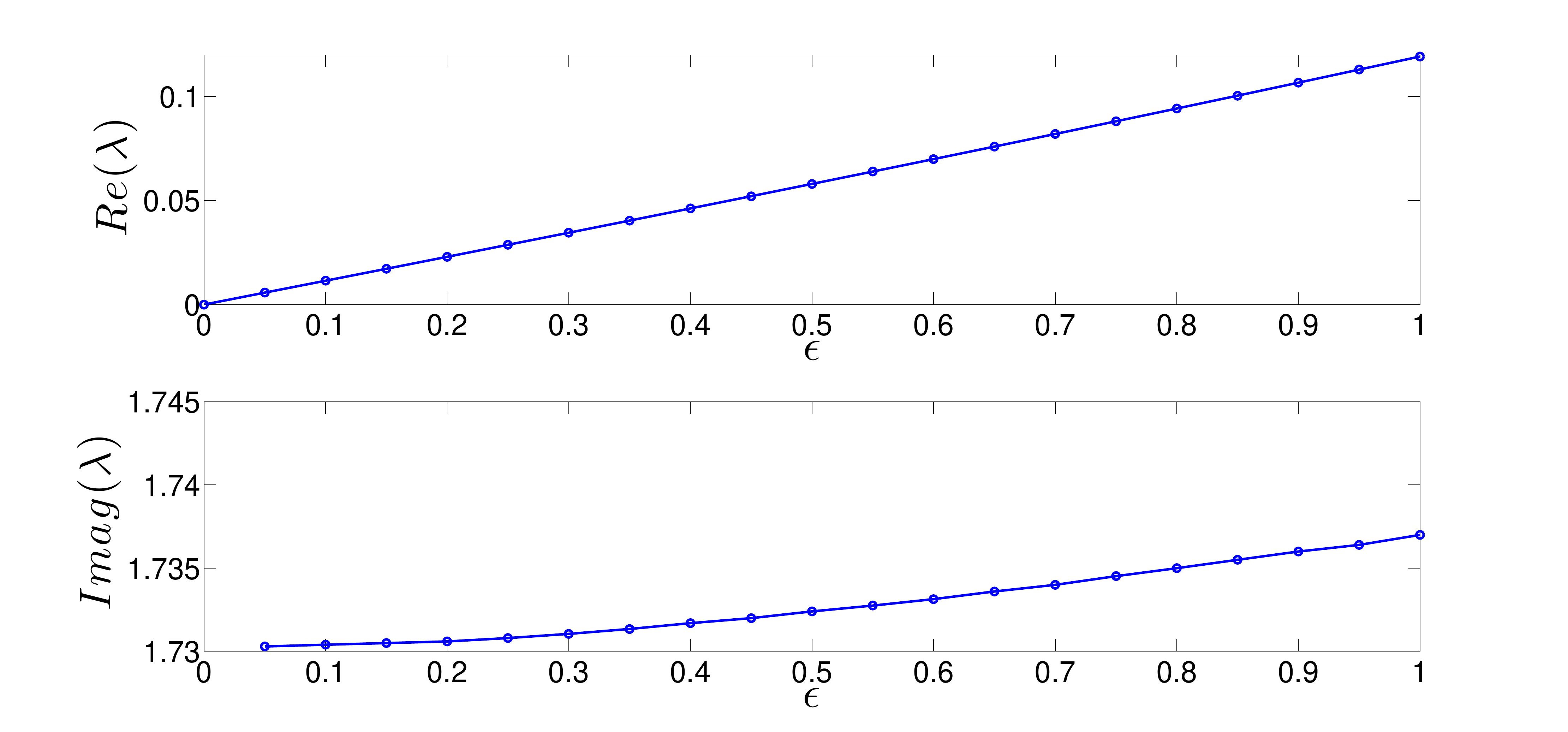}
\caption{The same as the previous figure but now for
the case of $x_0=2$, leading to (additional)
instability due to the bifurcation
of the relevant eigenvalue to the right half plane.}
\label{fig4}
\end{figure}

\section{Conclusions and Future Challenges}

In this work, we have introduced a $\mathcal{P T}$-symmetric variant of the
widely relevant (in mechanical, electrical and other physical systems --as
e.g. in particle physics--) Klein-Gordon field theory. We have done so by
incorporating an anti-symmetric dashpot term featuring gain on one side of
the domain and loss symmetrically on the other side. While such $\mathcal{P
T}$-symmetric Klein-Gordon settings are starting to be analyzed even
experimentally at the discrete level (especially in the case of few
nodes), the case of the continuum limit is virtually unexplored. We have seen
that for such systems, the stationary kink states of the model persist;
however, their spectral stability is dramatically affected. If the kinks of
interest are centered at the exact interface between the gain and loss, then
they preserve their spectral stability, a feature that our Theorem
\ref{thm:31} allows us to establish rigorously. If they are centered on the
lossy side, they become further stabilized with parts of the point spectrum
moving to the left half plane (either on the real axis or as complex
eigenvalue pairs). In this case, if there is a spectral instability, it is
weak, for the real part of the unstable eigenvalue will be
$\calO(\epsilon^3)$. On the other hand, by means of the symmetries of the
model, it is evident that when they are centered on the gain side of the
medium, then the relevant eigenvalues bifurcate to the right half plane,
giving rise to instability.

While the above lay some of the foundations of the subject of $\mathcal{P
T}$-symmetric Klein-Gordon field theories, there are numerous questions that
remain open both as regards the existence and in connection to the stability
properties of coherent structures in such models. A feature that we did not
touch upon, which our preliminary observations suggest that it is even
numerically quite technically challenging is the fate of the resonance in the
case of $x_0 \neq 0$. Understanding whether this edge bifurcation from the
continuous spectrum may contribute to an instability in such a setting would
be a topic of theoretical and numerical interest in its own right. In
addition, it would be interesting to develop an instability index theory for
such quadratic polynomials in which the linear coefficient is not definite.
As was seen in \cite{Kap3} in the case of quadratic matrix polynomials, if
the linear coefficient is definite, then the total number of eigenvalues with
positive real part is bounded above by the total number of negative
eigenvalues for the Schr\"odinger operator $\calH$. Unfortunately, the proof
of this result heavily depended upon the definiteness of the linear term, and
consequently it is not clear as to how it can be extended to the problem
considered in this paper.

%Furthermore, and at a more general level, it is interesting to explore
%quadratic polynomial pencils in the form of Eq.~\ref{23} and attempt to
%understand their mathematical properties for different types of
%anti-symmetric ``dashpots'' $\gamma(x)$.
%(even for $H$
%bearing a constant potential).

As noted earlier there are states which are genuinely time-dependent in
Klein-Gordon models, and these states include the traveling kinks and the
breathers. While one may be inclined to doubt the existence of genuinely
traveling kinks given the spatial inhomogeneity imposed by the gain-loss
profile, it would be interesting to explore the dynamics of such waves.  For
the breathers it is unclear whether the $\mathcal{P T}$-symmetric extension
will preserve their structure (and if so, under which conditions this may
occur). These questions are currently under consideration and will be
reported in future publications.

%\appendix

%\section{}
%\label{app:A}

%\begin{proof}(Lemma \ref{le:40})

%\end{proof}

\end{document}